\documentclass[preprint,12pt]{elsarticle}

\usepackage[utf8]{inputenc}
\usepackage[T1]{fontenc}
\usepackage{amsmath, amssymb, amsthm}
\usepackage{geometry}
\usepackage{hyperref}
\usepackage{float}     
\usepackage{algorithm}
\usepackage{algpseudocode}
\usepackage{microtype} 
\usepackage{booktabs}  
\usepackage{enumitem}  

\newtheorem{theorem}{Theorem}
\newtheorem{lemma}{Lemma}

\journal{Theoretical Computer Science}

\begin{document}

\begin{frontmatter}

\title{An $\mathcal{O}(\log N)$ Time Algorithm for the Generalized Egg Dropping Problem}

\author{Kleitos Papadopoulos\corref{cor1}}
\ead{kleitospa@gmail.com}
\ead[url]{https://orcid.org/0000-0002-7086-0335}

\cortext[cor1]{Corresponding author}

\begin{abstract}
The generalized egg dropping problem is a classic challenge in dynamic programming and sequential state-space reduction. Standard dynamic programming evaluates the minimax minimum number of tests in $\mathcal{O}(K \cdot N^2)$ time. A known approach formulates the testable thresholds as a partial sum of binomial coefficients and applies binary search to reduce the time complexity to $\mathcal{O}(K \log N)$. In this paper, we demonstrate that binary search over the complete sequential test domain is suboptimal. By restricting a binary search over multiples of $K$, we isolate a dynamic structural envelope that guarantees convergence. We prove that this boundary balances the search depth against the combinatorial evaluation cost, cancelling the $K$ variable to strictly bound the search phase to $\mathcal{O}(\log N)$. Combined with an incremental traversal, our algorithm eliminates the standard combinatorial bottlenecks. Furthermore, we formulate an explicit $\mathcal{O}(1)$ space policy to dynamically reconstruct the optimal decision tree. Finally, we establish a bounding lemma proving the optimal threshold is tightly confined within a strict radius, laying the theoretical groundwork for alternative $\mathcal{O}(K)$ algorithms driven by direct numerical approximation.
\end{abstract}


\end{frontmatter}

\section{Introduction}
\label{sec:intro}

The generalized egg dropping problem is a fundamental worst-case optimization problem utilized to study the efficiency of dynamic programming, minimax control policies, and sequential search bounds. Beyond its pedagogical value, this framework models fundamental challenges in state-space reduction, adversarial search bounds, and worst-case complexity analysis. While it historically models operations research scenarios including destructive material stress-testing and fault localization, its primary theoretical value lies in optimizing combinatorial mappings where probes are costly and failure events deplete a limited inventory. 

The problem is defined as follows: \textit{Given $K$ identical test items that break upon falling from a critical threshold floor $F$, and a linear domain of $N$ thresholds, what is the minimum number of tests $T^*$ required to identify $F$ under adversarial outcomes, and what is the optimal testing policy?}

Traditional solutions rely on Bellman's Principle of Optimality to construct a state-transition matrix \cite{bellman1957dynamic}. The standard dynamic programming approach requires $\mathcal{O}(K \cdot N^2)$ operations, optimizable to $\mathcal{O}(K \cdot N)$ via monotonicity tracking. For very large state spaces (e.g., $N \ge 10^{18}$), this matrix approach rapidly exceeds practical spatial and temporal limits.

A known combinatorial improvement relies on modeling the maximum threshold height $E(T, K)$ coverable by $T$ tests and $K$ items as a partial sum of the $T$-th row of Pascal's triangle \cite{boardman2004egg, konhauser1996which}. By combining this parameterization with a discrete binary search, the previous state-of-the-art yields an $\mathcal{O}(K \log N)$ time algorithm \cite{wills2009connections}. 

We introduce a bounded algebraic approach that circumvents the sequential combinatorial bottlenecks, isolating the optimal threshold in bounded $\mathcal{O}(\log N)$ time. In addition to the multiplier-search algorithm, we prove a second structural result showing that the optimum is confined to a radius of $K$ around $(N \cdot K!)^{1/K}$. This confinement can be implemented as an analytic algorithm: compute the integer floor of this $K$-th root, initialize near the lower edge of the resulting window, and traverse only the local interval using the same incremental identities. To summarize, our primary contributions are:
\begin{itemize}[noitemsep, topsep=4pt]
    \item A scalar bounding technique with a self-contained proof establishing a self-balancing $\mathcal{O}(\log N)$ search envelope.
    \item A structural state-update sequence leveraging reverse parity to compute the exact minimum threshold iteratively.
    \item A confinement lemma showing that a direct analytic implementation can recover $T^*$ by checking only a radius-$K$ interval around $(N \cdot K!)^{1/K}$, once the integer root is available.
    \item An $\mathcal{O}(1)$ space algebraic formulation to dynamically reconstruct the optimal testing policy, summarized alongside algorithm complexities in Table \ref{tab:complexity}.
\end{itemize}

\begin{table}[ht]
\centering
\caption{Worst-case complexities for the generalized egg dropping problem}
\label{tab:complexity}
\footnotesize
\setlength{\tabcolsep}{4pt}
\renewcommand{\arraystretch}{1.15}
\begin{tabular}{@{}lcc@{}}
\toprule
\textbf{Algorithm} & \textbf{Time} & \textbf{Space} \\ 
\midrule
Standard DP \cite{skiena2020algorithm} 
    & $\mathcal{O}(K N^2)$ 
    & $\mathcal{O}(K N)$ \\
Optimized DP 
    & $\mathcal{O}(K N)$ 
    & $\mathcal{O}(N)$ \\
Binomial binary search \cite{wills2009connections} 
    & $\mathcal{O}(K \log N)$ 
    & $\mathcal{O}(1)$ \\
\textbf{Proposed multiplier search} 
    & \textbf{$\mathcal{O}(\log N)$} 
    & \textbf{$\mathcal{O}(1)$} \\
\textbf{Proposed analytic approach}$^\dagger$ 
    & \textbf{$\mathcal{O}(K + \mathcal{R}(N,K))$} 
    & \textbf{$\mathcal{O}(1)$} \\ 
\midrule
\multicolumn{3}{@{}p{0.82\linewidth}@{}}{\scriptsize
$^\dagger$ Here, $\mathcal{R}(N,K)$ denotes the time required to compute the exact integer $K$-th root of $N \cdot K!$.} \\
\bottomrule
\end{tabular}
\end{table}
\section{Related Work}
\label{sec:related}

The mathematical formulation of the problem traces back to mathematical puzzle collections \cite{konhauser1996which} and was formally categorized by Boardman \cite{boardman2004egg}, who defined the ``Egg-Drop Numbers.'' In computer science, the problem is a ubiquitous exercise in Dynamic Programming, featured in standard texts like Skiena’s Algorithm Design Manual \cite{skiena2020algorithm}.

Combinatorial optimization typically approaches the inverse problem (finding the minimum number of tests $T^*$ given $N$ floors) through search techniques. Wills \cite{wills2009connections} outlined the parameterization combined with binary search that leads to an $\mathcal{O}(K \log N)$ time algorithm. Because the initial search space is bounded by $N$, a standard binary search maintains a depth of $\mathcal{O}(\log N)$, while computing the capacity at each node requires an $\mathcal{O}(K)$ loop. Under the worst-case limit where $K \approx \log_2 N$, this standard algorithmic complexity degrades to $\mathcal{O}(\log^2 N)$. 

In this paper, we first present an algorithm that circumvents standard evaluation bottlenecks by restricting a discrete binary search to scalar multiples of $K$. We demonstrate that this parameterization balances computational depth against the evaluation loop, yielding a strictly bounded $\mathcal{O}(\log N)$ time complexity. We subsequently establish a bounding lemma proving that the optimal solution is restricted to a very tight window. This theoretical confinement lays the groundwork for alternative numerical approaches.

\section{Mathematical Formulation}
\label{sec:math}

Let $E(T, K)$ denote the maximum number of floors testable under a minimax policy. This capacity is the partial sum of binomial coefficients \cite{boardman2004egg}:
\begin{equation}
    E(T, K) = \sum_{i=1}^{K} \binom{T}{i}.
\end{equation}

Our primary objective is to compute the inverse mapping for the optimal value:
\begin{equation}
    T^* = \min \left\{ T \in \mathbb{N} \;\middle|\; E(T, K) \ge N \right\}.
\end{equation}

By definition, $\binom{T}{i} = 0$ for all $i > T$. Thus, when items exceed tests ($K \ge T$), the summation captures the entire binomial row. By the standard Binomial Theorem, this simplifies to $E(T, K) = 2^T - 1$.

\section{The Algorithm}
\label{sec:algorithm}

We partition the calculation of the optimal test budget $T^*$ into three phases: a decision-tree lower bound check, a discrete root bounding with integrated state caching, and an incremental linear search. The procedure is formalized in Algorithm \ref{alg:egg}.

\subsection{Phase 1: Decision-Tree Lower Bound}
Every sequential binary test partitions the remaining search space into at most two outcomes. By the information-theoretic properties of binary decision trees, isolating one state among $N+1$ possible outcomes requires a minimum of $T_{\text{ideal}} = \lceil \log_2(N+1) \rceil$ tests. By definition, $2^{T_{\text{ideal}}} \ge N+1 > N$. If $K \ge T_{\text{ideal}}$, the item capacity constraint is non-binding, reducing the problem to a classical binary search. If the algorithm proceeds past this check, we establish a theoretical limit for the remainder of the execution: $K < \log_2(N+1)$.

\subsection{Phase 2: Discrete Root Bounding and Caching}
In the constrained regime where $K < T_{\text{ideal}}$, we require a sequence that guarantees structural capacity coverage. Because the capacity function evaluates a sum of strictly positive integers, it is bounded below by its largest single term ($E(T, K) \ge \binom{T}{K}$). The binomial term admits the following product expansion:
\begin{equation}
    \binom{T}{K} = \prod_{j=0}^{K-1} \frac{T-j}{K-j}.
\end{equation}
For every $0 \le j \le K-1$ and every valid $T \ge K$, we have
\begin{equation}
    \frac{T-j}{K-j} \ge \frac{T}{K},
\end{equation}
because $K(T-j) \ge T(K-j)$ is equivalent to $j(T-K) \ge 0$. Therefore,
\begin{equation}
    \binom{T}{K}
    = \prod_{j=0}^{K-1} \frac{T-j}{K-j}
    \ge \prod_{j=0}^{K-1} \frac{T}{K}
    = \left(\frac{T}{K}\right)^K.
\end{equation}

By parameterizing the sequence as scalar multiples of the item count ($T = K \cdot M$), we establish the algebraic lower bound:
\begin{equation} \label{eq:multiplier_bound}
    E(K \cdot M, K) \ge \left(\frac{K \cdot M}{K}\right)^K = M^K.
\end{equation}
To guarantee coverage over $N$, the capacity must reach $2^{T_{\text{ideal}}} > N$. Setting $M^K \ge 2^{T_{\text{ideal}}}$ and resolving for integer bounds dictates a maximum multiplier of $M_{\text{max}} = 2^{\lceil T_{\text{ideal}} / K \rceil}$. Therefore, limiting the discrete search domain strictly to $M \in [1, M_{\text{max}}]$ guarantees full coverage. 

We isolate the optimal multiplier $M^*$ utilizing binary search over this bounded interval. Within each step, computing the capacity $E(K \cdot M_{\text{mid}}, K)$ evaluates in an $\mathcal{O}(K)$ arithmetic loop. Standard binary search inherently identifies the greatest valid lower bound. By caching the evaluated capacity $E_{\text{mid}}$ and its trailing binomial term exclusively when $E_{\text{mid}} < N$, the algorithm completes Phase 2 holding the exact base state mapping for $T = K \cdot (M^*-1)$, seamlessly eliminating the need for a secondary initialization loop.

\subsection{Phase 3: Incremental Linear Search}
The algorithm initializes the final sequential tracking at $T_{\text{cached}} = K \cdot (M^* - 1)$ using the cached capacity $E$ and boundary coefficient $B = \binom{T_{\text{cached}}}{K}$. To traverse the remaining local search gap, we advance sequentially utilizing Pascal's identity $\binom{T+1}{i} = \binom{T}{i} + \binom{T}{i-1}$. Summing this identity across the limits yields an arithmetic progression:
\begin{equation}
    E(T+1, K) = 2 \cdot E(T, K) - \binom{T}{K} + 1.
\end{equation}

To dynamically update the trailing boundary coefficient $B$, we algebraically refactor the binomial expansion:
\begin{equation} \label{eq:binomial_update}
    \binom{T+1}{K} = \binom{T}{K} + \frac{\binom{T}{K} \cdot K}{T+1-K}.
\end{equation}

Because this expression evaluates to exactly $\binom{T}{K-1}$, and the algorithm logic maintains $T \ge K$, the denominator safely avoids division by zero, guaranteeing an exact integer division. Furthermore, because the loop terminates when $E \ge N$, the intermediate numerator $B \times K$ remains bounded by $N \times K$.

\begin{algorithm}[H] 
\caption{Analytic Minimax Search}
\label{alg:egg}
\begin{algorithmic}[1]
\Require Test items $K \ge 1$, States $N \ge 1$
\Ensure Minimum worst-case tests $T^*$

\If{$N \le 1$ \textbf{or} $K = 1$} 
    \State \textbf{return} $N$ 
\EndIf

\Statex
\State \Comment{\textbf{Phase 1: Decision-Tree Lower Bound}}
\State $T_{\text{ideal}} \gets \lceil \log_2(N + 1) \rceil$
\If{$K \ge T_{\text{ideal}}$} 
    \State \textbf{return} $T_{\text{ideal}}$ 
\EndIf

\Statex
\State \Comment{\textbf{Phase 2: Discrete Root Bounding and Caching}}
\State $L \gets 1, \; R \gets 2^{\lceil T_{\text{ideal}} / K \rceil}$ \hfill \Comment{Evaluates natively via bitwise left-shift}
\State $E_{\text{cached}} \gets 0, \; B_{\text{cached}} \gets 1, \; T_{\text{cached}} \gets K$

\While{$L < R$} 
    \State $M_{\text{mid}} \gets \lfloor (L + R) / 2 \rfloor$
    \State $E_{\text{mid}} \gets 0, \; \text{term} \gets 1, \; T_{\text{mid}} \gets K \times M_{\text{mid}}$
    \For{$i = 1$ \textbf{to} $K$}
        \State $\text{term} \gets \text{term} \times (T_{\text{mid}} - i + 1) / i$ \hfill \Comment{Exact integer division}
        \State $E_{\text{mid}} \gets E_{\text{mid}} + \text{term}$
        \If{$E_{\text{mid}} \ge N$} 
            \State \textbf{break} 
        \EndIf
    \EndFor
    
    \If{$E_{\text{mid}} \ge N$} 
        \State $R \gets M_{\text{mid}}$ 
    \Else 
        \State $L \gets M_{\text{mid}} + 1$ 
        \State $E_{\text{cached}} \gets E_{\text{mid}}$ \hfill \Comment{Cache the greatest valid lower bound}
        \State $B_{\text{cached}} \gets \text{term}, \; T_{\text{cached}} \gets T_{\text{mid}}$ 
    \EndIf
\EndWhile

\Statex
\State \Comment{\textbf{Phase 3: Incremental Linear Search}}
\State $E \gets E_{\text{cached}}, \; B \gets B_{\text{cached}}, \; T \gets T_{\text{cached}}$
\While{$E < N$}
    \State $E \gets 2E - B + 1$
    \State $B \gets B + (B \times K) / (T + 1 - K)$ \hfill \Comment{Exact integer division}
    \State $T \gets T + 1$
\EndWhile

\State \textbf{return} $T$
\end{algorithmic}
\end{algorithm}

\section{Complexity Analysis}
\label{sec:complexity}

\begin{lemma} \label{lemma:gap}
The discrete root bounding in Phase 2 isolates the optimal threshold $T^*$ to an interval of size at most $K$, such that $T_{\text{cached}} < T^* \le T_{\text{cached}} + K$.
\end{lemma}
\begin{proof}
Phase 2 parameterizes the search sequence as $T = K \cdot M$. Because the capacity function $E(T, K)$ increases monotonically, standard binary search zeroes in on the optimal bounding integer $L$ such that $E(K \cdot L, K) \ge N$ and $E(K \cdot (L-1), K) < N$. 

By definition, $T^*$ is the minimum integer satisfying $E(T^*, K) \ge N$. Therefore, the optimal threshold is strictly bounded between these evaluated numbers: $K \cdot (L-1) < T^* \le K \cdot L$. The remaining gap evaluates to precisely $K \cdot L - K \cdot (L - 1) = K$. Initializing the sequential search at the cached state bounds the linear search phase to at most $K$ iterations.
\end{proof}

\begin{theorem}
Algorithm \ref{alg:egg} computes the exact target $T^*$ in an upper bound of $\mathcal{O}(\log N)$ time and $\mathcal{O}(1)$ space.
\end{theorem}
\begin{proof}
Phase 1 operates in $\mathcal{O}(1)$ time. In Phase 2, the algorithm searches the mathematically restricted domain $M \in [1, M_{\text{max}}]$. This bounded interval acts as a self-balancing algorithmic envelope. The binary search isolates the multiplier $M^*$ in exactly $\log_2(M_{\text{max}}) = \lceil T_{\text{ideal}} / K \rceil$ geometric splits. Because Phase 1 defines $T_{\text{ideal}} \approx \log_2 N$, the search depth strictly evaluates to an asymptotic upper bound of $\mathcal{O}(\frac{\log N}{K})$. 

Within each search split, the algorithm evaluates $E(K \cdot M_{\text{mid}}, K)$ via an internal $\mathcal{O}(K)$ loop. Consequently, the variable $K$ algebraically cancels out between the geometric search depth and the arithmetic step cost:
\begin{equation}
    \text{Phase 2 Time} = \mathcal{O}\left(\frac{\log N}{K}\right) \times \mathcal{O}(K) = \mathcal{O}(\log N)
\end{equation}
This locks the worst-case binary search computation to $\mathcal{O}(\log N)$, circumventing the standard $\mathcal{O}(K \log N)$ bottleneck. Furthermore, configuring the upper bound natively requires only a single $\mathcal{O}(1)$ bitwise left-shift.

By Lemma \ref{lemma:gap}, the linear search in Phase 3 is restricted to a maximum of $K$ discrete tracking increments. Each sequential step updates using $\mathcal{O}(1)$ arithmetic bounds. Because Phase 1 structurally enforces $K < \log_2(N+1)$, the cumulative sum of the phases, $\mathcal{O}(\log N) + \mathcal{O}(K)$, condenses to a final bound of $\mathcal{O}(\log N)$.
\end{proof}

\section{Algebraic Confinement of the Optimal Threshold}
\label{sec:confinement}

We prove that the exact integer threshold $T^*$ is locked within a tight geometric radius. This result exposes a fundamental inefficiency of unconstrained searches and hints at an algorithmic alternative: the possibility of bypassing combinatorial search entirely in favor of numerical techniques.

\begin{lemma} \label{lemma:algebraic_bound}
The exact optimal test threshold $T^*$ is bounded within a radius of $K$ from $x = (N \cdot K!)^{\frac{1}{K}}$, satisfying:
\begin{equation}
    x - K < T^* < x + K.
\end{equation}
\end{lemma}

\begin{proof}
Assume the constrained operational regime where items do not exceed the required tests ($K < T^*$). The maximum capacity for exactly $K$ drops evaluates to $E(K, K) = 2^K - 1$. Since $K$ drops are insufficient to cover $N$, the optimal target $T^*$ must be strictly greater than $K$, ensuring $T^* - 1 \ge K$.

By definition of the optimal threshold, $E(T^*-1, K) < N$. Because $T^* - 1 \ge K$, we can apply the standard falling factorial lower bound $\binom{X}{K} \ge \frac{(X-K+1)^K}{K!}$ to the trailing binomial term, establishing:
\begin{align}
    N &> E(T^*-1, K) \nonumber \\
      &\ge \binom{T^*-1}{K} \nonumber \\
      &\ge \frac{(T^* - 1 - K + 1)^K}{K!} = \frac{(T^* - K)^K}{K!}. \label{eq:continuous_lower}
\end{align}
Isolating $T^*$ from Inequality \eqref{eq:continuous_lower} yields the strict lower bound:
\begin{equation}
    (N \cdot K!)^{\frac{1}{K}} > T^* - K.
\end{equation}

Conversely, by definition, $N \le E(T^*, K)$. Utilizing the Vandermonde identity for combinatorial convolution, we can bound the capacity strictly from above:
\begin{align}
    \binom{T^* + K}{K} &= \sum_{j=0}^K \binom{T^*}{j} \binom{K}{K-j} \nonumber \\
    &\ge 1 + \sum_{j=1}^K \binom{T^*}{j} \nonumber \\
    &> E(T^*, K) \ge N.
\end{align}
Because all binomial coefficients are naturally bounded by $\binom{X}{K} \le \frac{X^K}{K!}$, substitution yields:
\begin{equation}
    N < \binom{T^* + K}{K} \le \frac{(T^*+K)^K}{K!} \implies (N \cdot K!)^{\frac{1}{K}} < T^* + K.
\end{equation}

Letting $x = (N \cdot K!)^{\frac{1}{K}}$, the union of these inequalities rigorously confirms that the exact integer solution is locked within $x - K < T^* < x + K$.
\end{proof}

Lemma \ref{lemma:algebraic_bound} reveals a critical structural property of the generalized egg dropping problem: the seemingly massive sequential search space bounded by $N$ is practically an illusion. The true search radius is dictated almost entirely by an offset bounded by $K$. 

This confinement forms the basis of an alternative algorithmic approach. By bypassing the binary search entirely to compute the integer floor $r = \lfloor (N \cdot K!)^{\frac{1}{K}} \rfloor$, one can initialize at the lower edge $T_0 = \max\{K, r-K\}$, compute the initial pair $E(T_0,K)$ and $\binom{T_0}{K}$ using a single $\mathcal{O}(K)$ binomial pass, and then apply the same incremental update identities from Phase 3 until $E(T,K) \ge N$. Since Lemma \ref{lemma:algebraic_bound} guarantees $T^* < x+K \le r+K+1$, the scan is confined to a window of at most $2K+1$ candidate values. This allows us to determine the optimal threshold in $\mathcal{O}(K + \mathcal{R}(N,K))$ time, where $\mathcal{R}(N,K)$ is the time complexity of computing the exact integer $K$-th root of $N \cdot K!$.

\section{Retracing the Optimal Policy}
\label{sec:policy}

Computing the minimax budget $T^*$ establishes the necessary operational bounds, but the decision-maker must also generate the sequential testing policy. By utilizing the combinatorial bounds, the optimal decision tree can be reconstructed dynamically without spatial transition matrices.

Let $F_{\text{safe}}$ denote the highest floor known to be safe, and let $F_{\text{break}}$ denote the lowest floor proven to break the item. Initially, $F_{\text{safe}} = 0$ and $F_{\text{break}} = N+1$. Suppose it is at an operational state with $k$ items and $t$ optimal tests remaining. 

The minimax optimization dictates that if a queried threshold breaks the item, the target must be isolated within the untested floors strictly below the query utilizing $k-1$ items and $t-1$ tests. This sub-domain capacity is bounded by $E(t-1, k-1)$. Therefore, relative to the highest known safe floor $F_{\text{safe}}$, the optimal threshold to query next is:
\begin{equation}
    f_{\text{test}} = \min\Big(F_{\text{break}} - 1, \; F_{\text{safe}} + E(t-1, k-1) + 1 \Big).
\end{equation}

Rather than calculating this capacity dynamically using $\mathcal{O}(K)$ summations at each branch, we track the terminal states backward in $\mathcal{O}(1)$ time per step. Algorithm \ref{alg:egg} concludes leaving the final state capacity $E = E(t,k)$ and boundary coefficient $B = \binom{t}{k}$ cached in memory. First, partition the boundary coefficient into its breaking and surviving branches:
\begin{equation}
    B_{\text{break}} = \binom{t-1}{k-1} = B \frac{k}{t}, \quad B_{\text{stay}} = \binom{t-1}{k} = B - B_{\text{break}}.
\end{equation}
Because combinatorial coefficients are integers, evaluating $B \times k/t$ guarantees an exact integer division.

Substituting the forward recurrence $E(t, k) = 2E(t-1, k) - B_{\text{stay}} + 1$, we explicitly isolate the structural capacities of the surviving and breaking sub-domains:
\begin{equation}
    E_{\text{stay}} = E(t-1, k) = \frac{E + B_{\text{stay}} - 1}{2},
\end{equation}
\begin{equation}
    E_{\text{break}} = E(t-1, k-1) = E - E_{\text{stay}} - 1.
\end{equation}
Because combinatorial floor limits adhere to strict binomial parity, $(E + B_{\text{stay}} - 1)$ inherently forms an even integer, ensuring the halving step guarantees an exact whole number.

Upon executing the test at floor $f_{\text{test}}$, the policy branches based on the adversarial outcome, updating the tracking variables dynamically in constant time:
\begin{itemize}
    \item \textbf{Item breaks:} The upper bound collapses ($F_{\text{break}} \gets f_{\text{test}}$). The algebraic state tracks the breaking sub-domain ($E \gets E_{\text{break}}$, $B \gets B_{\text{break}}$). Resources decrement to $t-1$ tests and $k-1$ items.
    \item \textbf{Item survives:} The lower bound rises ($F_{\text{safe}} \gets f_{\text{test}}$). The algebraic state tracks the surviving sub-domain ($E \gets E_{\text{stay}}$, $B \gets B_{\text{stay}}$). Resources decrement to $t-1$ tests, and $k$ remains unchanged.
\end{itemize}

Should the remaining active interval fall beneath unconstrained limits ($k \ge \lceil \log_2(F_{\text{break}} - F_{\text{safe}}) \rceil$), the policy defaults to standard binary search midpoints. By replacing evaluation loops with basic arithmetic, test execution resolves in exactly $\mathcal{O}(T^*)$ time per run. Furthermore, mapping the full execution tree across all $N$ potential thresholds processes in optimal $\mathcal{O}(N)$ global time without invoking array caches.

\section{Computational Results}
\label{sec:experiments}

To empirically validate the theoretical time complexities, we evaluated the proposed algorithm against the state-of-the-art (SOTA) binomial dynamic programming approach utilizing binary search \cite{wills2009connections}. We also include the analytic algorithm suggested by Lemma \ref{lemma:algebraic_bound}, which computes the integer root $\lfloor (N \cdot K!)^{1/K} \rfloor$ and scans the confined radius using the same incremental recurrence. All algorithms were implemented in Python, leveraging its native arbitrary-precision integer arithmetic to safely handle massive state spaces without overflow artifacts. 

To mitigate operating system jitter and interpreter initialization overhead at the microsecond scale, execution times were measured using a rigorous benchmarking framework. Functions were warmed up to initialize memory caches, and execution times were averaged across 10,000 iterations per state. Table \ref{tab:experiments} summarizes the results across a spectrum of problem scales, ranging from trivial base cases to extreme combinatorial bounds ($N = 10^{300}$).

\begin{table}[ht]
\centering
\caption{Empirical Execution Times: SOTA vs. Proposed Algorithms}
\label{tab:experiments}
\footnotesize
\setlength{\tabcolsep}{3.5pt}
\renewcommand{\arraystretch}{1.15}
\begin{tabular}{@{}lcrrrr@{}}
\toprule
\textbf{$N$} & \textbf{$K$} & \textbf{$T^*$} & 
\textbf{SOTA ($\mu$s)} & 
\textbf{$\log N$ Algorithm ($\mu$s)} & 
\textbf{Analytic Method ($\mu$s)} \\ 
\midrule
$1$ & $1$ & $1$ & $0.03$ & $0.04$ & $0.17$ \\
$100$ & $2$ & $14$ & $1.61$ & $1.15$ & $1.60$ \\
$1000$ & $15$ & $10$ & $4.50$ & $0.11$ & $0.27$ \\
$10^4$ & $3$ & $40$ & $3.67$ & $1.83$ & $2.88$ \\
$10^9$ & $5$ & $166$ & $11.77$ & $3.12$ & $4.73$ \\
$10^{18}$ & $10$ & $290$ & $37.27$ & $8.12$ & $11.53$ \\
$10^{30}$ & $20$ & $272$ & $69.94$ & $11.63$ & $22.57$ \\
$10^{100}$ & $50$ & $1972$ & $355.64$ & $44.93$ & $52.59$ \\
$10^{300}$ & $40$ & $498644626$ & $1896.42$ & $184.12$ & $74.37$ \\ 
\bottomrule
\end{tabular}
\end{table}
The empirical results strongly corroborate the asymptotic analysis established in Section \ref{sec:complexity}. At the base case ($N=1, K=1$), the SOTA and $\mathcal{O}(\log N)$ algorithms perform almost identically, while the analytic method has slightly higher constant overhead in trivial cases. The analytic method has slightly higher constant overhead in trivial cases.

However, the advantages of the proposed algorithm's core search logic become immediately apparent as we scale beyond trivial tree depths. Consider the parameter set $N=10^{18}, K=10$. Because the item inventory is strictly less than the information-theoretic lower bound ($10 < \lceil \log_2(10^{18}+1) \rceil$), the $\mathcal{O}(1)$ early-exit condition is bypassed, forcing a complete combinatorial evaluation. Here, the state-of-the-art approach requires $37.27\ \mu\text{s}$ to resolve the problem. In contrast, the proposed algorithm computes the exact threshold in just $8.12\ \mu\text{s}$, while the analytic method resolves the same instance in $11.53\ \mu\text{s}$. This validates the efficiency of isolating the search envelope to multiples of $K$ (Phase 2), traversing the remaining gap incrementally (Phase 3), and alternatively exploiting the confinement radius from Lemma \ref{lemma:algebraic_bound}.

As both the operational threshold $N$ and the item inventory $K$ continue to scale into massive magnitudes, the asymptotic divergence becomes pronounced. For large problem instances (e.g., $N=10^{100}$), the SOTA execution time degrades due to the repeated $\mathcal{O}(K)$ inner evaluation loops required at each node of its binary search tree.

Conversely, the proposed algorithm's execution times remain light. This confirms the strict $\mathcal{O}(\log N)$ bound, where the arithmetic influence of $K$ is minimal during the discrete root bounding phase. At the extreme limit of $N=10^{300}$, the multiplier-search approach achieves an order of magnitude speedup over the prior algorithm, while the analytic implementation is even faster in this benchmark because the confinement lemma reduces the final verification to a narrow local scan after exact root extraction.

\section{Open Problems and Theoretical Limits}
\label{sec:open_problems}

While our $\mathcal{O}(\log N)$ algorithm resolves the standard computational bottlenecks, Lemma \ref{lemma:algebraic_bound} reveals an even deeper property: the optimal threshold $T^*$ is constrained by a very tight mathematical radius, almost entirely independent of the massive overall search space $N$. This insight prompts an open question about the theoretical limits of the problem.

Under the standard word RAM model, is it theoretically possible to compute the exact minimax testing threshold $T^*$ for the generalized egg dropping problem in $\mathcal{O}(K)$ time?

The conceptual pathway to this bound is illuminated by the confinement radius. By proving that $T^*$ is strictly locked within the interval $x - K < T^* < x + K$, where $x = (N \cdot K!)^{\frac{1}{K}}$, the necessity of searching the global domain is eliminated. If the integer floor of this $K$-th root can be evaluated efficiently, an algorithm could simply traverse this localized interval utilizing the constant-time arithmetic progressions outlined in Phase 3 of Algorithm \ref{alg:egg}, yielding a strict $\mathcal{O}(K)$ global time complexity.

Determining whether this is achievable theoretically requires resolving this algorithmic constraint: precise root extraction. One must guarantee an integer floor for the $K$-th root of an astronomically large integer $N$ without introducing hidden bitwise complexity factors under the standard word RAM assumption.

\section{Conclusion}
\label{sec:conclusion}

We advance the state-of-the-art for the generalized egg dropping problem by introducing a novel algorithm that isolates the minimax test threshold in $\mathcal{O}(\log N)$ time. Furthermore we proved that the optimal solution is confined to a tight geometric radius and suggested an alternative analytic approach. Finally, we provide an explicit $\mathcal{O}(1)$ space framework to dynamically map the optimal testing policy on the fly.

\section*{Data and Code Availability}

The Python implementation and benchmarking framework used to generate the computational results in Section \ref{sec:experiments} are available upon request, and are designed to be fully reproducible under standard execution environments.


\end{document}